\documentclass[onecolumn,pra,superscriptaddress]{revtex4-1}

\usepackage[dvips]{graphicx}
\usepackage{amsmath,amssymb,amsthm,mathrsfs,amsfonts,dsfont}
\usepackage{subfigure, epsfig}
\usepackage{braket}
\usepackage{bm}
\usepackage{enumerate}
\usepackage{color}
\usepackage{framed}
\usepackage{epstopdf}

\newcommand{\ketbra}[2]{\mbox{$\left| #1 \right\rangle   \left\langle #2 \right|$}}

\newcommand{\tr}{\mathrm{Tr}}

\newtheorem{lemma}{Lemma}
\newtheorem{theorem}{Theorem}
\newtheorem{definition}{Definition}

\begin{document}
%\preprint{APS/123-QED}
%\title{Coherence based source-independent quantum random number generation}
%\title{Source-independent quantum random number generation via measuring coherence}
\title{Coherence as a resource for source-independent quantum random-number generation}

\author{Jiajun Ma}
%\email{xiao.yuan.ph@gmail.com}
\affiliation{Center for Quantum Information, Institute for Interdisciplinary Information Sciences, Tsinghua University, Beijing 100084, China}

\author{Aishwarya Hakande}
%\email{xiao.yuan.ph@gmail.com}
\affiliation{Center for Quantum Information, Institute for Interdisciplinary Information Sciences, Tsinghua University, Beijing 100084, China}

\author{Xiao Yuan}
\email{xiao.yuan.ph@gmail.com}
\affiliation{Department of Materials, University of Oxford, Parks Road, Oxford OX1 3PH, United Kingdom}
\affiliation{Center for Quantum Information, Institute for Interdisciplinary Information Sciences, Tsinghua University, Beijing 100084, China}

\author{Xiongfeng Ma}
\email{xma@tsinghua.edu.cn}
\affiliation{Center for Quantum Information, Institute for Interdisciplinary Information Sciences, Tsinghua University, Beijing 100084, China}

%\author{Jiajun Ma, Xiao Yuan, Aishwarya Hakande, Xiongfeng Ma$^{\dag}$}
%\address{Center for Quantum Information, Institute for Interdisciplinary Information Sciences, Tsinghua University, Beijing 100084, China}
%\ead{$^\dag$xma@tsinghua.edu.cn}

\begin{abstract}
%Measuring quantum states provides a means to generate genuine random numbers. It has been shown that genuine randomness can be obtained even with an uncharacterized randomness source [Phys.~Rev.~X 6, 011020 (2016)]. In this work, we propose a framework that formalizes the idea of realizing source-independent quantum random number generation via measuring coherence. Without full state tomography, the coherence of the source can be estimated by coherence witnesses. The existing uncertainty-relation-based schemes can be treated as special cases under the coherence framework, where we design a nonlinear coherence witness that can essentially yield the same results. Meanwhile, we propose a source-independent random number generation scheme, which can achieve a higher randomness generation rate than the uncertainty-relation-based ones. Our work sheds light on the close relation between the resource framework of coherence and the operational task of random number generation.

Measuring quantum states provides a means to generate genuine random numbers. It has been shown that genuine randomness can be obtained even with an uncharacterized source by measuring two incompatible bases [Phys.~Rev.~X 6, 011020 (2016)]. As coherence is the necessary source for generating randomness, we extend the scheme and propose a framework for quantum random number generation with general uncharacterized coherence resource. The previous scheme can be treated as a special case under the framework by considering a nonlinear uncertainty-relation-based coherence witness.
Considering general coherence witnesses, we propose a source-independent random-number generation scheme that achieves a higher randomness generation rate. Our paper highlights the close relation between coherence and random number generation, and may shed light on designing general semi-device-independent quantum information processing protocols.
\end{abstract}

%\pacs{}

\maketitle

\section{Introduction}
Random number generation has many important applications in various tasks. In some cases, such as Monte Carlo simulation, it only requires the random numbers to be statistically unbiased. Pseudo random numbers or physical random numbers based on classical mechanics, such as coin flipping and noise measuring, are sufficient. The outcomes of these procedures may appear random, but they are in principle predictable. In cryptography, one of the security foundations lies on the unpredictability of random numbers. For instance, a cryptophytic key requires genuinely random bits. The random numbers via classical mechanic procedures are not suitable for cryptosystems. Therefore, it is important to study the generation of unpredictable (or genuine) random numbers.

%Most common procedures of generating random numbers, e.g., rolling a dice or flipping a coin, are governed by the deterministic classical mechanics. the main requirement lies on how to make sure that the random numbers are statistically unbiased. In these cases, physical random numbers based on classical mechanics or even pseudo random numbers are sufficient.

According to Born's rule \cite{born1926quantentheorie}, the measurement outcome of a quantum system that is in the superposition of the measurement basis is unpredictable. Based on quantum mechanic principles, there are many quantum random number generation (QRNG) schemes proposed in the past two decades. For a review of the subject, one can refer to Refs.~\cite{Ma2016QRNG,herrero2017quantum} and the references therein. In general, a QRNG setup consists of two parts, a source that contains randomness and a measurement that reads out the randomness. As shown in Fig.~\ref{fig:RNG}, the source emits a sequence of quantum signals, and the readout system measures them to produce random outcomes. For instance, if the states from the source form a sequence of qubits  $\ket{+} = (\ket{0}+\ket{1})/\sqrt{2}$, and the measurement is a projection onto the $\{\ket{0},\ket{1}\}$ basis, the outcomes are genuinely random.

\begin{figure}[bht]
\centering
\resizebox{7cm}{!}{\includegraphics[scale=1]{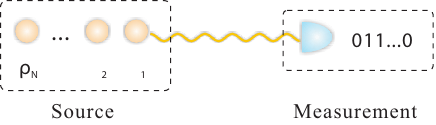}}
\caption{A quantum random-number generator can be generally decomposed into two parts: source and measurement.}\label{fig:RNG}
\end{figure}

In practice, the source may contain both genuine randomness and classical noise. The latter can normally be influenced by some unexpected environment parameters, such as the temperature. In a cryptographic picture, an adversary Eve may take advantage of the classical noise. Consider a source that emits maximally mixed state $\rho = 1/2(\ket{0}\bra{0}+\ket{1}\bra{1})$. Although the measurement outcomes on the $Z$ basis appear random, they are not genuinely random. This can be understood in the presence of Eve, who simply prepares $\ket{0}$ or $\ket{1}$ based on a predetermined random string. In this case, the measurement outcomes are entirely predictable. In fact, the predictability in random-number generators becomes a major security issue in current cryptosystems \cite{NSA}.

The key job of randomness analysis is to quantify the genuine randomness so that it can be extracted. One way of randomness quantification relies on modeling the QRNG implementation, which normally requires calibrating the source and measurement devices \cite{Ma2013Extractor}. In practice, however, the calibration is often hard to perform thoroughly. Once implemented devices deviate from theoretical models, the randomness can be compromised. Several proposals have been proposed to solve this problem. For instance, QRNGs via certifying randomness based on nonlocality tests \cite{pironio2010random} are called self-testing or device-independent schemes where neither the calibration of the quantum source nor the measurement is required. In practice, realizing a loophole-free Bell test is experimentally challenging, which requires high-fidelity state preparation and detection efficiency, as well as the accurate chronological sequence control \cite{hensen2015loophole}. Furthermore, the randomness generation rate is quite limited owing to a small violation of Bell inequalities obtained with the state-of-the-art technology \cite{miller2014universal,arnon2016simple}.

In real-life applications, a fully device-independent scenario may be too restrictive. By putting certain reasonable assumptions to devices, the performance of QRNGs would become practically acceptable. Along this direction, tremendous efforts have been devoted to find a trade-off between device independence and high randomness generation rate \cite{cao2015loss,cao2016source,Marangon17SI,Bischof2017measure,brask2017megahertz,van2017semi}. In this paper, we focus on the scenario of generating genuine randomness with well-calibrated measurement devices but uncharacterized sources, namely, source-independent quantum random-number generation (SIQRNG) \cite{cao2016source}. With certain reasonable assumptions on measurement devices, the SIQRNG schemes can be very practical. For instance, with a continuous-variable system, the randomness generation rate in such a scenario has achieved the gigabits per second (Gbps) regime \cite{Marangon17SI}. The intuition behind these schemes is the quantum uncertainty relation. Given two complementary measurement bases, $X$ and $Z$, if the outcome uncertainty of the $X$-basis measurement is small, its uncertainty of the $Z$-basis measurement must be large. The information on the complementary basis can be used to reveal the genuine randomness within the source. Since the uncertainty relation is state independent, these QRNG schemes are source independent.

Recently, genuine randomness has been shown to be essentially related to the coherence of the input quantum state in the measurement basis  \cite{Yuan15,yuan2016interplay}. There, the source is trusted and the extractable randomness is quantified in the asymptotic limit. In this paper, we propose a framework that extends this idea to a more general setting where the source is uncharacterized and the measurement outcomes are finite. We show that SIQRNG can be realized via measuring the coherence of the input quantum states. Moreover, we propose a method of estimating the coherence via coherence witnesses. By designing a nonlinear coherence witness, we show that the coherence estimation yields the same randomness quantification with the previous uncertainty-relation-based SIQRNG schemes. This coherence witness picture sheds light on the fact that the uncertainty-relation-based SIQRNG does not maximally efficiently extract the genuine randomness from the source. Our randomness analysis is based on the assumption that the source emits signals that are independent and identically distributed (i.i.d.). Furthermore, we propose a new SIQRNG scheme that uses tomography to estimate coherence and thereby generally achieves a higher randomness generation rate than the uncertainty-relation-based ones.

The paper is organized as follows. In Sec.~\ref{Sec:Preliminaries}, we review the preliminaries on coherence measures and SIQRNG. Then, in Sec.~\ref{Sec:Coherence}, we introduce a witness to measure the coherence. In Sec.~\ref{Sec:Framework}, we present the framework of QRNG via measuring coherence, based on which, a SIQRNG protocol is proposed in Sec.~\ref{Sec:protocol}. In Sec.~\ref{Sec:Simulation}, with numerical simulations, we show that our protocol generally achieves a higher randomness generation rate than the uncertainty-relation-based schemes. Finally, we conclude in Sec.~\ref{Sec:Discussion} with discussions on interesting related perspective subjects.

\section{Preliminaries: Coherence and SIQRNG}\label{Sec:Preliminaries}
\subsection{Resource theory of coherence}
The resource theory of coherence formalizes the intuition that quantum superposition is non-classical \cite{Baumgratz14}.  In this framework, with an orthogonal basis $\Pi=\{\ket{i}\}$ as the reference basis, one can define an incoherent state as $\sigma=\sum_ip_i\ketbra{i}{i}$, with $p_i\ge 0$ and $\sum_ip_i=1$. Incoherent operations are physical realizable operations that map incoherent states to incoherent states. Specifically, they are formed by the Kraus operators, $\{K_i\}$, that satisfy $K_i\mathcal{I}K_i^{\dagger}\in\mathcal{I}$, where $\mathcal{I}$ is the set of incoherent states.

Under this resource framework, a coherence measure needs to fulfill a few criteria. There are many proposals for coherence measures, such as the $l_1$-norm of coherence \cite{Baumgratz14}, the coherence of formation \cite{Yuan15}, and the robustness of coherence \cite{napoli2016robustness}. In this paper, we adopt the relative entropy of coherence \cite{Baumgratz14},
\begin{equation}\label{Eq:relCoherence}
C(\rho)=H(\Delta_{\Pi}(\rho))-H(\rho),
\end{equation}
where $H(\rho)$ is the von Neumann entropy of $\rho$ and $\Delta_{\Pi}(\rho)\equiv\sum_i\ketbra{i}{i}\rho\ketbra{i}{i}$ is the dephasing operation in the reference basis $\Pi$. Recently, the relative entropy of coherence is linked to the genuine randomness obtained by measuring $\rho$ in the basis $\Pi$ \cite{yuan2016interplay}.

\subsection{SIQRNG} \label{Sub:SIQRNG}
In the framework of SIQRNG, Eve prepares a bipartite quantum state of systems $A$ and $E$, represented by $\tau_{AE}$, where each system consists of $N$ partitions. Considering the most general attack by Eve, the joint state of the $2N$ partitions $\tau_{AE}$ can be prepared in an arbitrary bipartite state, where each partition may have arbitrary dimension, including the dimension of 0, and the joint state can be entangled among partitions and between $A$ and $E$. In this case, $\tau_{AE}$ can represent any quantum state with arbitrary dimension. Then Eve sends system $A$ to the legitimate user Alice and retains the rest system $E$. After receiving $A$, Alice measures each of the $N$ partitions using a random measurement setting. From Alice's perspective, as illustrated in Fig.~\ref{fig:SIQRNG2}, the source effectively emits a sequence of quantum states $\{\rho_i\}$, where $\rho_i$ denotes the reduced state of partition $i\in\{1, 2, ..., N\}$. Note that, in general, different partitions $\{\rho_i\}$ can be entangled with each other as the joint state $\tau_{AE}$ is generally entangled. Alice tries to extract genuine randomness from the measurement outcomes. Meanwhile, Eve aims at predicting the random numbers extracted from $\tau_A$ with the assistance of $\tau_E$, where $\tau_A = \tr_E[\tau_{AE}]$ and $\tau_E = \tr_A[\tau_{AE}]$ are the reduced density matrices of $A$ and $E$, respectively.

\begin{figure}[htb!]
\centering
\resizebox{7cm}{!}{\includegraphics[scale=1]{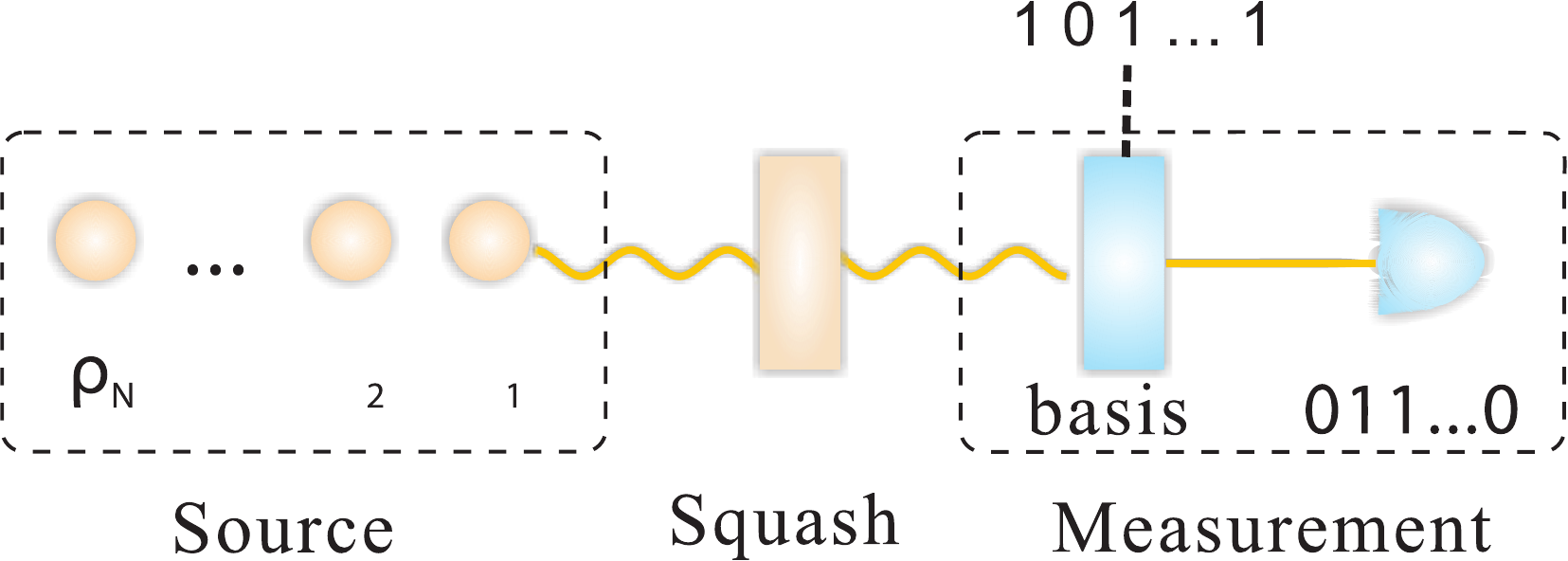}}
\caption{Source-independent quantum random-number generator. The squashing model is applied to transform the input states into qubits and vacuum states. Then, a basis is chosen to measure the squashed qubits. The dashed line on the measurement side denotes that, in some schemes \cite{cao2016source}, initial random bits are needed to choose the measurement basis randomly. Note that in some other schemes \cite{Marangon17SI}, a beam splitter is used to randomly select the measurement basis, which, of course, might raise the question why one should trust a beam splitter. Such kinds of discussions appeared in Ref.~\cite{cao2016source}.}\label{fig:SIQRNG2}
\end{figure}

In this paper, we consider a common optical realization of the SIQRNG, as most of the practical QRNGs are implemented with quantum optics \cite{Ma2016QRNG,herrero2017quantum}. In the measurement setting, we assume Alice uses typical optical components, such as phase modulators and threshold detectors, and she employs random assignment for double-click events. This is a widely used detection model with optical implementation.

In the source-independent scenario, since the source is assumed to be controlled by Eve, $\tau_A$ can be prepared in an arbitrary dimensional Hilbert space. As a major component of the security analysis for SIQRNG, we employ the \emph{squashing model} to remove the dimension arbitrariness of $\tau_A$ \cite{GLLP2004,Ma2007Entangled}. In the squashing model, Alice first applies a squashing operation which projects the received state to a qubit or a vacuum state and then performs the qubit measurement. In literature \cite{beaudry2008squashing,tsurumaru2008security}, the squashing model has been proven to be equivalent to the common optical implementations using threshold detectors and proper postprocessing. Note that the squashing model can be applied here since in the source-independent scenario, the measurement device is assumed to be trusted, thus it can be composed of threshold detectors that fit into the squashing model. Then, in the following randomness analysis, we directly employ this squashing model. In the postprocessing, the squashing model requires Alice to randomly assign the measurement outcome to 0 or 1 for the double-click events, which would affect the net generated randomness. This random assignment issue will be further analyzed in Sec.~\ref{Sec:Double}.

%In the source-independent scenario, as the measurement device is trusted, we assume that it consists of two threshold single-photon detectors, whose click events represent measurement outcome 0 and 1. The threshold detector model is also widely employed in the security analysis of QKD \cite{Ma2007Entangled}. Following the results in literature \cite{beaudry2008squashing,tsurumaru2008security}, once threshold single-photon detectors are used and proper postprocessing on the measurement outcomes are performed, one can apply the squashing model to the security analysis of SIQRNG. In the squashing model, Alice first applies a squashing operation which projects the received state to a qubit or a vacuum state, and then performs the qubit measurement. This is a virtual procedure, which is proven to be equivalent to the practical measurement on the arbitrary dimensional state $\tau_A$, provided the threshold detector model and a proper postprocessing on the measurement outcomes.

After the squashing operation, Alice can project state $\tau_A$ to $n$ qubits and $N-n$ vacuum states. Alice performs qubit measurements on the $n$ qubits. The rest $N-n$ vacuum states can be regarded as measurement losses. Since measurement devices are trusted in the SIQRNG scenario, one can assume that the $n$ qubits are fair sampled from the $N$ states \cite{cao2016source}. That is, the detection efficiency loophole here is not considered. Here, we remark that our analysis can be applied to other cases, such as the one where Alice uses a photon non-demolition measurement (compatible with the squashing model). In practice, post-selection (discarding losses) is feasible, whereas photon non-demolishing measurement technology is not available with current technologies. If Alice uses post selection, she might need initial randomness in making measurement basis choices. However, in the asymptotic limit, the amount of this consumed initial randomness can be reduced to a negligible ratio. In fact, for the scheme proposed in Table~\ref{Table:protocol}, Alice can choose $q\rightarrow0$, and this amount of randomness is sufficient for randomness extraction. The following randomness analysis will focus on the $n$ squashed qubits.

Here, we briefly review the previous SIQRNG scheme \cite{cao2016source}.
\begin{enumerate}[1)]
\item
An untrusted source (controlled by Eve) emits a sequence of quantum states.
\item
Alice (or Eve) squashes the quantum states into qubits and vacua. Alice discards the vacua and retains the $n$ squashed qubits.
\item
Alice randomly chooses $n_x$ qubits out of the $n$ squashed qubits and measures them in the $X$ basis. Within $n_x$ outcomes, the ratio of outcome $\ket{-}$ is $e_{xb}$, which is defined to be the error rate. Ideally, Alice expects the source to emit state $\ket{+}$ and hence the result of $\ket{-}$ is defined as an error.
\item
Alice measures the rest $n_z=n-n_x$ squashed qubits in the $Z$ basis to obtain $n_z$ raw random bits.
\item 
Alice extracts  $n_z(1-S(e_{xb}+\theta))-t_e$ random bits from the raw random bits using a universal hashing function, where $S$ is the binary Shannon entropy, $\theta$ is the deviation due to statistical fluctuations, and $2^{-t_e}$ is the failure probability of the randomness extraction.
\end{enumerate}

In above scheme, the security proof techniques of QKD are employed in the randomness analysis. Here is the argument. Ideally, the source should emit states $\ket{+}$. Then, Alice measures them in the $Z$ basis $\{\ket{0},\ket{1}\}$ to generate random numbers. In the scenario of SIQRNG, the source is allowed to emit arbitrary quantum states in arbitrary dimensions. On the measurement end, one can consider a virtual protocol. First, a squashing model is applied to project the quantum states into a sequence of qubit and vacuum states. Then Alice performs an error correction procedure that transforms all states to $\ket{+}$. Finally she measures them in the $Z$ basis. By designing the error correcting code appropriately, this operation can be commuting with the $Z$-basis measurement \cite{shor2000simple}. Also the squashing model is proven to be equivalent to the threshold detection model with appropriate postprocessing \cite{beaudry2008squashing}. Hence, Alice does not need to perform the virtual protocol, which requires a universal quantum computer. Instead, she can perform the $Z$-basis measurement first and then apply a randomness extraction on the measurement outcomes. The number of extractable random bits, $n_z(1-S(e_{xb}+\theta))-t_e$, is derived from the uncertainty principle together with the $X$-basis error rate $e_{xb}$. The randomness extraction essentially functions the same as the error correction procedure in the virtual protocol.

In the randomness extraction, or the error correction in the virtual protocol, Alice needs to know the error rate in the $X$ basis, $\{\ket{+},\ket{-}\}$, where $\ket{\pm}=(\ket{0}\pm\ket{1})/\sqrt{2}$. Thus, in the SIQRNG scheme, Alice needs to randomly test the quantum state in the $X$ basis to estimate its error rate, which is later used for randomness quantification \cite{cao2016source}.

In this paper, we would show the randomness of SIQRNG from a difference perspective. The idea is based on the recently discovered link between coherence and genuine randomness in a quantum state \cite{Yuan15,yuan2016interplay}. Instead of applying error correction, Alice estimates the coherence of the qubit states via measuring randomly sampled qubit states. Then she measures the rest qubits in the $Z$ basis to generate the raw randomness. The estimated coherence is used to bound the genuine randomness of the raw data. In postprocessing, she distills the genuine randomness from the raw data.

\section{Measuring coherence with a coherence witness}\label{Sec:Coherence}
Normally, to estimate the coherence of an unknown state, one needs to perform a full state tomography to obtain the density matrix $\rho$. In some cases, full tomography information is not available. For example, as the dimension of the state increases, the number of required measurement in tomography increases quadratically, which becomes challenging for experiments. Thus, it is interesting to estimate the coherence of an unknown state without a full tomography.

A similar problem raises in the field of entanglement measure, where it is expected to estimate the entanglement with a limited number of measurements. The solution there is to employ entanglement witnesses, which are originally designed to justify whether quantum states are entangled or not to estimate entanglement \cite{audenaert2006correlations,eisert2007quantitative,guhne2007estimating}. Recently, this idea has been extended to the resource theory of coherence  \cite{napoli2016robustness}, i.e., estimating the coherence with coherence witnesses.

%\subsection{Coherence witness}
The original coherence witness is a linear function of the density matrix $\rho$ \cite{napoli2016robustness}. Here, we extend this notion to an arbitrary function of $\rho$.
\begin{definition}
A coherence witness is a function of $\rho$, $W(\rho)$, that satisfies the following criteria,
\begin{enumerate}
\item
$\forall\rho\in\mathcal{I}$, $W(\rho)\ge0$;

\item
$\exists\rho\notin\mathcal{I}$, $W(\rho) < 0$.
\end{enumerate}
\end{definition}

A linear coherence witness has been shown to be useful to bound the coherence \cite{napoli2016robustness}. Here, we design a nonlinear coherence witness to bound the relative entropy of coherence.

\begin{lemma}
Given a reference basis $\Pi=\{\ket{i}\}$ in a $d$-dimensional Hilbert space,
\begin{equation}\label{Eq:CWu}
W_u(\rho)=H(\Delta_\Xi(\rho))-\log_2d,
\end{equation}
is an coherence witness, where $\Xi$ and $\Pi$ are mutually unbiased bases of the same Hilbert space, so that $\Xi$ is maximally incompatible with $\Pi$.
\end{lemma}

\begin{proof}
The incoherent state set $\mathcal{I}$ is defined in basis $\Pi$. Since $\Xi$ is a mutually unbiased basis of $\Pi$, for all $\rho\in\mathcal{I}$, we have $H(\Delta_\Xi(\rho))=\log_2d$, and hence $W_u(\rho)=0$. Also, for all $\ket{i'}\in\Xi$, we have $H(\Delta_\Xi(\ket{i'}\bra{i'}))=0$, and hence $W_u(\ket{i'}\bra{i'})=-\log_2 d<0$.
\end{proof}

\begin{theorem} \label{Thm:CW2Coh}
Given a reference basis $\Pi=\{\ket{i}\}$ and a state $\rho$ in a $d$-dimensional Hilbert space, the relative entropy of coherence $C(\rho)$ can be bounded by the coherence witness $W_u(\rho)$ defined in Eq.~\eqref{Eq:CWu},
\begin{equation}\label{Eq:LBC}
C(\rho)\ge -W_u(\rho)=\log_2d-H(\Delta_{\Xi}(\rho)),
\end{equation}
where $\Xi$ is a mutually unbiased basis of $\Pi$.
\end{theorem}

\begin{proof}
The dephasing operators, $\Delta_\Pi(\rho)$ and $\Delta_\Xi(\rho)$, can be viewed as two projection measurements, which have  the quantum uncertainty relation \cite{berta2010uncertainty},
\begin{equation}\label{Eq:uncertainty}
H(\Delta_\Pi(\rho))+H(\Delta_\Xi(\rho))\ge -\log_2 c+H(\rho),
\end{equation}
where $c=\max_{i,i'}|\left\langle i | i' \right\rangle|^2$, with $\ket{i}\in\Pi$ and $\ket{i'}\in\Xi$.

The two bases $\Pi$ and $\Xi$ are mutually unbiased, and hence $c=1/d$. Rearranging the terms in Eq.~\eqref{Eq:uncertainty} and using the definition $C(\rho)=H(\Delta_\Pi(\rho))-H(\rho)$, Eq.~\eqref{Eq:LBC} is obtained.
%From results in Ref.~\cite{berta2010uncertainty}, there is a quantum uncertainty relation $C_x(\rho)+C_z(\rho)\ge 1-H(\rho)$, where $C_x(\rho)$ is the relative entropy of coherence when the reference basis is $X$ basis, while $C_z(\rho)$ is the relative entropy of coherence when the reference basis is $Z$ basis. Inserting $C_x(\rho)=H(\Delta_X(\rho))-H(\rho)$ and $C_z(\rho)=H(\Delta_Z(\rho))-H(\rho)$, Eq.~\eqref{Eq:LBC} is obtained.
\end{proof}

From Theorem \ref{Thm:CW2Coh}, one can estimate the relative entropy of coherence via measuring the state in the complementary basis of the reference basis. This idea is similar to the one employed in the uncertainty-relation-based SIQRNG \cite{cao2016source,Marangon17SI}. There, the intrinsic randomness generated via the $Z$-basis measurement is estimated by measuring the state in the complementary $X$ basis. In the next section, we propose a framework that formalizes the relation between these two scenarios.

\section{Framework of SIQRNG via measuring coherence}\label{Sec:Framework}
The task of estimating coherence of an unknown quantum state shares similarities with randomness evaluation in SIQRNG. In both scenarios, the source state is uncharacterized whereas the measurement is trusted. Meanwhile, the amount of genuine randomness within the source can be quantified by the coherence of the quantum state  \cite{Yuan15,yuan2016interplay}. Therefore, extracting genuine randomness in SIQRNG can be reduced to the problem of estimating coherence within the source. In this section, we would present a framework that links the two tasks.

%The task of estimating coherence of an unknown quantum state shares similarities with randomness evaluation in SIQRNG. In both scenarios, the source state is uncharacterized while the measurement is trusted. One difference is that the dimension of the state is known in measuring coherence, while it is not given in SIQRNG. Nevertheless, as discussed above, this discrepancy can be resolved by applying the squashing model to SIQRNG, such that the quantum states are squashed into a sequence of qubits. Finally, it was noted that the amount of genuine randomness within the source can be quantified by the coherence within the quantum state  \cite{Yuan15,yuan2016interplay}. Therefore, extracting genuine randomness in SIQRNG can be reduced to the problem of estimating coherence within the source. In this section, we would present an framework that links the two tasks.

\subsection{Quantification of randomness}
Following the discussion of SIQRNG in Sec.~\ref{Sec:Preliminaries}, we focus on the $n$ squashed qubits, which contribute one raw data bit each. Alice needs to quantify the genuine randomness in the $n$-bit raw data from the $n$-qubit state, $\tau_A=\textrm{Tr}_E(\tau_{AE})\in \mathcal{H}_2^{\otimes n}$, where $\mathcal{H}_2$ denotes a two-dimensional Hilbert space. In the partial trace, we put the $N-n$ vacuum states to system $E$. Note that the $n$ qubits can be correlated with each other, or even with Eve's system  $\tau_E$.

Suppose Alice randomly chooses $n_z$ qubits and measures them in the $Z$ basis to generate raw random bits, whereas she measures the rest $n-n_z$ qubits in some other complementary bases for parameter estimation, which would give Alice information about $\tau_A$. Denote the measurement outcome in the $Z$ basis (an $n_z$-bit string) by $K_z$. Here, Alice's measurement can be viewed as a dephasing operation on each qubit of subsystem $A$ in the $Z$ basis, $\Delta_{Z^{\otimes n_z}}^A(\tau_{AE})$. Then the randomness contained in $K_z$ is quantified by \cite{konig2009operational}
\begin{equation}\label{Eq:RZN}
R^{\varepsilon_1}(K_z) = \min_{\tau_{AE}}H_{\mathrm{min}}^{\varepsilon_1}(A|E)_{\Delta_{Z^{\otimes n_z}}^A(\tau_{AE})},
\end{equation}
where the minimization runs over all possible states of Eve that satisfy $\textrm{Tr}_E(\tau_{AE})=\tau_A$, and $H_{\mathrm{min}}^{\varepsilon_1}$ is the smooth min-entropy, defined in Appendix \ref{Ap:Conditional}, with a smooth parameter $\varepsilon_1$.

The min-entropy $R^{\varepsilon_1}(K_z)$ is the key parameter for randomness extraction. With universal hashing  \cite{impagliazzo1989pseudo}, such as Teoplitz-matrix hashing, one can extract random bits that are $\varepsilon$-close to a uniformly distributed string from Eve's point of view. Here, the security parameter is $\varepsilon=\varepsilon_1+\varepsilon_2$, with $\varepsilon_2$ as the failure probability introduced in the randomness extraction procedure.

%Based on two universal hashing \cite{impagliazzo1989pseudo} and the randomness estimation $R^{\varepsilon_1}(K_z)$, one can extract random bits that are $\varepsilon$ close to a uniformly distributed string which is independent on $E$, where $\varepsilon=\varepsilon_1+\varepsilon_2$, with $\varepsilon_2$ the failure probability induced by the randomness extraction procedure.

%\begin{equation}\label{Eq:RZN2}
%l^{\varepsilon} = R^{\varepsilon_1}(K_z) - 2\log_2\frac{1}{\varepsilon_2} + 1, \,\mathrm{with}\,\varepsilon = \varepsilon_1 + \varepsilon_2.
%\end{equation}
%In randomness extraction, additional random seeds are required. Thanks to the leftover hash lemma \cite{impagliazzo1989pseudo}, there exits strong randomness extractors, Toeplitz hashing extractor for example, whose output random string can be independent of the random seeds. Therefore, the random seeds can be reused.

\subsection{Randomness analysis}
In the following, we analyze the randomness with the assumption that the $n$ squashed qubits are i.i.d. This assumption is also made in the scenario of collective attacks in QKD, where Eve attacks each signal in an i.i.d.~manner. In a more general setting, the $n$ qubits might be entangled, which corresponds to the scenario of coherent attacks in QKD. It is proven that the security parameter for coherent attacks is only polynomially larger than the security parameter for collective attacks \cite{christandl2009postselection}. The extra information available to the adversary for coherent attacks can be compensated by slightly reducing the size of the final random bits in the privacy amplification stage. We expect that a similar argument can be employed here to obtain the security proof against the most general sources. We leave the randomness analysis with an correlated source for future study.

%The security against the most general source will be discussed at the end of this section.

%More generally, when the source is not i.i.d, the eavesdropper is allowed to prepare the sequence of quantum states to be an arbitrary joint quantum state. Note that the i.i.d. source corresponds to the collective attacks in the field of quantum key distribution, while the non-i.i.d. source corresponds to the coherent attacks there. It is proven that the security parameter for coherent attacks is only polynomially larger than the security parameter for collective attacks  \cite{christandl2009postselection}. The extra information available to the adversary for coherent attacks can be compensated by slightly reducing the size of the final random bits in the privacy amplification stage. Therefore, similar argument can be employed here to obtain the security proof against non-i.i.d. sources.

With the i.i.d.~assumption, the joint state that outputs the raw random bits can be expressed by $\tau_{AE}=\rho_{AE}^{\otimes n_z}$, where $\rho_{AE}$ is the squashed joint quantum state of each signal. From Eq.~\eqref{Eq:RZN}, one can have
\begin{equation}\label{Eq:RHM}
R^{\varepsilon_1}(K_z) = \min_{\rho_{AE}}H_{\mathrm{min}}^{\varepsilon_1}(A|E)_{\Delta_{Z^{\otimes n_z}}^A(\rho_{AE}^{\otimes n_z})}.
\end{equation}
For $n_z\ge \frac85\log_2\frac{2}{\varepsilon_1^2}$, the smooth min-entropy can be lower bounded by the conditional von Neumann entropy, defined as $H(A|E)_{\rho_{AE}}=H(\rho_{AE})-H(\rho_E)$ \cite{tomamichel2009fully}. Thus one has
\begin{equation}\label{Eq:RNC}
R^{\varepsilon_1}(K_z) \ge n_z\min_{\rho_{AE}}H(A|E)_{\Delta_Z^A(\rho_{AE})}-7.09\sqrt{n_z\log_2\frac{2}{\varepsilon_1^2}},
\end{equation}
whose derivation is shown in Appendix \ref{Ap:Derivation}. Meanwhile, $\min_{\rho_{AE}}H(A|E)_{\Delta_Z^A(\rho_{AE})}$ is related to the relative entropy of coherence of $\rho_A$ \cite{yuan2016interplay},
\begin{equation}\label{Eq:EC}
\min_{\rho_{AE}}H(A|E)_{\Delta_Z^A(\rho_{AE})}=C(\rho_A),
\end{equation}
where the reference basis for $C(\rho_A)$ is the $Z$ basis. Inserting Eq.~\eqref{Eq:EC} into Eq.~\eqref{Eq:RNC}, one has
\begin{equation}\label{Eq:RNC2}
R^{\varepsilon_1}(K_z)\ge n_zC(\rho_A)-7.09\sqrt{n_z\log_2\frac{2}{\varepsilon_1^2}}.
\end{equation}
%Inserting Eq.~\eqref{Eq:RNC2} into Eq.~\eqref{Eq:RZN2}, the number of extractable random bits is given by
%\begin{equation}\label{Eq:RRB}
%l^{\varepsilon}\ge n_zC(\rho_A)-c\sqrt{n_z\log_2\frac{2}{\varepsilon_1^2}}- 2\log_2\frac{1}{\varepsilon_2} + 1.
%\end{equation}
We remark that this expression only holds for $n_z\ge\frac85\log_2\frac{2}{\varepsilon_1^2}$, which is normally satisfied in practice, typically, $n_z\ge 95$ for $\varepsilon_1=10^{-10}$.

\subsection{Randomness analysis via a coherence witness}
The randomness analysis result Eq.~\eqref{Eq:RNC2} implies that one can estimate the amount of randomness in a quantum state by measuring the coherence of the state. In Sec.~\ref{Sec:Coherence}, we have shown that, without full state tomography, one can lower bound the coherence with coherence witnesses. Therefore, there is a close relation between the coherence witness and the SIQRNG: Any coherence witness that is able to lower bound the coherence can be employed to realize a SIQRNG scheme.

As an example, we apply this analysis method to the SIQRNG scheme described in Sec.~\ref{Sub:SIQRNG}. The measurement used by Alice in the SIQRNG scheme forms a coherence witness $W_u$ as shown in Eq.~\eqref{Eq:CWu}. Then, applying Theorem \ref{Thm:CW2Coh} to Eq.~\eqref{Eq:RNC2}, the number of genuine random bits, denoted by $R^{\varepsilon_1}_u$, is estimated by
\begin{equation}\label{Eq:RRU}
R^{\varepsilon_1}_u\ge -n_z W_u(\rho_A)-7.09\sqrt{n_z\log_2\frac{2}{\varepsilon_1^2}},
\end{equation}
where $W_u(\rho_A)=H(\Delta_X(\rho_A)) - 1$. In the asymptotic limit, when the number of emitted qubits $n$ approaches infinitely large, the randomness generation rate is given by
\begin{equation}\label{Eq:UCGR}
r_u = \frac{R^{\varepsilon_1}_u}{N}|_{N\rightarrow\infty}\ge q\beta(1-H(\Delta_X(\rho_A))),
\end{equation}
where $\beta = n/N$ is the transmittance of the signal, and $q=n_z/n$ is the ratio of $Z$-basis measurement. In this limit, Alice can set $q\rightarrow1$ to maximize $r_u$. Note that Eq.~\eqref{Eq:UCGR} coincides with the randomness generation rate evaluated via the complementary uncertainty relation \cite{cao2016source}.

To summarize, the security analysis in our framework is divided by two steps. First is the squashing operation, which maps uncharacterized signal states to qubit states. Applying the squashing model requires proper postprocessing of measurement outcomes, such as random assignment of double clicks to be discussed in Section \ref{Sec:Double}. Second is the coherence characterization of the squashed qubit states, which quantifies the secure random bits of the statistics obtained by measuring the qubit states.

\section{Tomography-based SIQRNG} \label{Sec:protocol}
Note that the SIQRNG protocols based on the complementary uncertainty relation do not necessarily extract the maximal amount of the randomness from the source. For instance, suppose that the source emits state $(\ket{0}+i\ket{1})/\sqrt{2}$, from Eq.~\eqref{Eq:UCGR}, one obtains a lower bound of $r_u$ to be 0, thus no genuine randomness can be extracted. Nevertheless, the measurement outcome on the $Z$ basis is in fact genuinely random. In this case, the genuine randomness cannot be revealed by the coherence witness using the $X$ measurement. Instead, the randomness can be witnessed with another complementary basis $Y = \{\ket{\pm i} = (\ket{0}\pm i\ket{1})/\sqrt{2}\}$. Without \emph{a priori} knowledge of the source, one might choose a bad witness to underestimate the genuine randomness.

For the SIQRNG scheme described in Sec.~\ref{Sub:SIQRNG}, by adding one more measurement basis, Alice can obtain a better estimation of the coherence of $\rho_A$ via a full state tomography. Then she can extract more genuine randomness from the raw data. Based on this observation, we propose a SIQRNG protocol based on tomography as presented in Table~\ref{Table:protocol}. %We would see that the generation rate of the random bits would be significantly improved for most states.

Note that in Table~\ref{Table:protocol} the $Z$-basis measurement data are used for both measurement tomography and randomness generation. We remark that the data used for tomography is, in principle, kept secret from any other party, which means the testing data are not revealed to the eavesdropper. Therefore, in principle, the $X$- and $Y$-basis measurement data can also be used to extract extra randomness with the analysis similar to that of the $Z$-basis measurement data. But in the limit where $q\rightarrow0$, the amounts of randomness generated by the $X$- and $Y$-basis measurement become negligible.

\begin{table}[htb]
\caption{Source-independent quantum random number generation} \label{Table:protocol}
\begin{framed}
\centering
\begin{enumerate}[1.]

\item
\textbf{State preparation}
\begin{enumerate}
\item
An untrusted source (might be controlled by Eve) emits $N$ quantum states in arbitrary dimensions, which are sequentially sent to the readout system.
\item
A \emph{squashing operation} transforms the quantum states into $n$ qubits and $N-n$ vacua.
\end{enumerate}

\item
\textbf{Measurement}
\begin{enumerate}
\item
The $N-n$ vacua are discarded and the remaining $n$ squashed qubit states are post-selected for randomness generation.
\item
Alice randomly chooses $n_x$, $n_y$, and $n_z$ qubits for the $X$-, $Y$-, and $Z$-basis measurements, with probability $q_x=q$, $q_y=q$, $q_z=1-2q$, respectively. Denote $p_x$, $p_y$, and $p_z$ to be the rates to obtain the outcomes $\ket{+}$, $\ket{i+}$, and $\ket{0}$, respectively.
\item
The $Z$-basis measurement outcomes are recorded as the raw data.
%\item
%A random trit $h\in\{0,1,2\}$ is chosen with probabilities of $q$, $q$, and $1-2q$, respectively.
%\item
%If $h = 0$, the qubit is measured in the $X$ basis. $P_x$ ($M_x$) denotes the number of outcome $\ket{+}$ ($\ket{-}$).
%\item
%If $h = 1$, the qubit is measured in the $Y$ basis. $P_y$ ($M_y$) denotes the number of outcome $\ket{i+}$ ($\ket{i-}$) detections.
%\item
%If $h = 2$, the qubit is measured in the $Z$ basis. $P_z$ ($M_z$) denotes the number of outcome $\ket{0}$ ($\ket{1}$) detections.
\end{enumerate}

\item
\textbf{State tomography} 
With the measurement results, $p_x$, $p_y$, and $p_z$, Alice can estimate the density matrix of the squashed qubits, $\rho_A$. Note that statistical fluctuations need to be considered here.

\item
\textbf{Randomness evaluation and extraction} 
With the information of $\rho_A$, Alice can bound the genuine randomness of the raw data and apply a proper randomness extractor to obtain the final random bits.
\end{enumerate}
\end{framed}
\end{table}

In the protocol, the coherence of the source $C(\rho_A)$ can be accurately estimated with a full tomography of $\rho_A$. Then, the number of genuine random bits, denoted by $R^{\varepsilon_1}_t$, can be estimated via Eq.~\eqref{Eq:RNC2},
\begin{equation}\label{Eq:RRB3}
R^{\varepsilon_1}_t\ge n_zC(\rho_A)-7.09\sqrt{n_z\log_2\frac{2}{\varepsilon_1^2}}.
\end{equation}
In the asymptotic limit, the randomness generation rate, denoted by $r_t$, is given by
\begin{equation}\label{Eq:TMGR}
r_t\ge \frac{R^{\varepsilon_1}_t}{N}|_{N\rightarrow\infty}=q_z\beta C(\rho_A),
\end{equation}
where $\beta=n/N$ is the transmittance of the signal and $q_z=n_z/n$ is the ratio of $Z$-basis measurement. In large data-size limit, Alice can set $q_z\rightarrow1$.

In tomography, the density matrix, $\rho_A$ can be estimated from measurement outcomes, $p_x$, $p_y$, and $p_z$, defined in Table \ref{Table:protocol}. Write $\rho_A$ as $\rho_A = (I + (2\vec{p}-1)\cdot\vec{\sigma})/2$, where $\vec{p}=(p_x, p_y, p_z)$ and $\vec{\sigma} = (\sigma_x, \sigma_y, \sigma_z)$  are the Pauli matrices. Substituting $\rho_A$ into Eq.~\eqref{Eq:relCoherence} and \eqref{Eq:RRB3}, one can get
\begin{equation}\label{Eq:Rtomo}
C(\rho_A) = H(p_z)- H\left(\frac{p_o+1}{2}\right),
\end{equation}
and
\begin{equation}\label{Eq:RH}
	R^{\varepsilon_1}_t \ge n_zH(p_z) - n_zH\left(\frac{p_o+1}{2}\right) -7.09\sqrt{n_z\log_2\frac{2}{\varepsilon_1^2}},
\end{equation}
with $p_o= \sqrt{4(p_x^2+p_y^2+p_z^2-p_x-p_y-p_z)+3}$.

\subsection{Squashing model}\label{Sec:Squashing}
As discussed in Sec.~\ref{Sec:Preliminaries}, the squashing model should be applied to project the quantum state from the uncharacterized source into a sequence of qubits. The squashing model requires that, depending on the measurement setting, appropriate postprocessing should be implemented. Here, Alice can employ the postprocessing of the squashing model used in the analysis of quantum state tomography  \cite{fung2011universal}.

According to the Supplemental Material of Ref.~\cite{fung2011universal}, the double click events should be considered to derive a set of bounds of the obtained statistics. Specifically, for each measurement bases $j\in\{X, Y, Z\}$, let $n^{0}_j$, $n^{1}_j$ and $n^{d}_j$ denote the numbers of the two single-click events and the double-click event, respectively. Then Alice obtains a set of qubit states, $\mathcal{S}$, that is compatible with the measurement results. That is, they fulfill the following constraints,
\begin{equation} \label{Eq:Squash}
\frac{n^{0}_j}{n^{0}_j+n^{1}_j+n^{d}_j}\le p_j \le \frac{n^{0}_j+n^{d}_j}{n^{0}_j+n^{1}_j+n^{d}_j}.%\\
%\frac{n^{0}_y}{n^{0}_y+n^{1}_y+n^{d}_y}\le p_y &\le \frac{n^{0}_y+n^{d}_y}{n^{0}_y+n^{1}_y+n^{d}_y},\label{Eq:SquashY}\\
%\frac{n^{0}_z}{n^{0}_z+n^{1}_z+n^{d}_z}\le p_z &\le \frac{n^{0}_z+n^{d}_z}{n^{0}_z+n^{1}_z+n^{d}_z}.\label{Eq:SquashZ}
\end{equation}
In the following, we denote the lower bound for $p_j$ in the above inequalities by $p_j^L$, and the upper bound by $p_j^U$.

In our protocol, Alice needs to consider the worst case of $\rho_A\in\mathcal{S}$. That is, she should minimize $C(\rho_A)$ in Eq.~\eqref{Eq:RH} over $\mathcal{S}$. In Appendix \ref{Ap:Partial}, we show that $C(\rho_A)$ is a unimodal function with respect to each $p_j$, with the minimal value achieved for $p_j=1/2$. Without loss of generality, from now on, we assume $p_j^U \ge 1/2$, otherwise Alice can flip the bit label in the $j$ basis. Denote the worst-case value of $p_j$ for the coherence quantification by $p^w_j$, thus $p_j^w=\max(p_j^L,1/2)$.

\subsection{Double clicks}\label{Sec:Double}
As discussed in Secs.~\ref{Sub:SIQRNG} and Sec.~\ref{Sec:Squashing}, the squashing model requires the random assignment of double-click events. Note that for the double-click events in the $X$ and $Y$ bases, the random assignment postprocessing need not be actually implemented as these measurement outcomes are only used for tomography testing in Eq.~\eqref{Eq:Squash} where Alice only needs to count the number of double-click events and evaluate the errors introduced by these events. On the other hand, as the measurement outcome in the $Z$ basis is used to generate the raw random bits, Alice should implement the random assignment on the double-click events to map them to single-value outcomes. In the rest of this subsection, we first introduce the random assignment method as directly required by the squashing model. Then we introduce an alternative discard method, which is more practical in experiments.

\subsubsection{Random assignment method}
Here, we consider a postprocessing method that Alice randomly assigns 0 or 1 to the double-click events in the $Z$ basis. After the squashing model analysis, Alice obtains $p^w_z$ as the worst case estimation of $p_z$. Then, in the random assignment procedure, the probability of assigning value $0$ to the double-click events, denoted by $p_a$, should be compatible with $p^w_z$,
\begin{equation}
p^w_z=\frac{n_z^0+p_an_z^d}{n_z^0+n_z^1+n_z^d}.
\end{equation}
Thus, $p_a$ is given by
\begin{equation}\label{Eq:pa}
p_a=\frac{p^w_zn_z-n_z^0}{n_z^d}.
\end{equation}
Note that the random assignment method generally consumes extra randomness, which should be taken into account when evaluating the net randomness generation rate. Here, the randomness cost in the double-click assignment procedure is $n_z^dH(p_a)$. Thus the asymptotic net randomness generation rate is given by
\begin{equation}\label{Eq:AssRtomo}
	\begin{aligned}
		R^{n}_{t}=&R^{\varepsilon_1}_t-n_z^dH(p_a) \\
\ge & n_zH(p_z^w) - n_zH\left(\frac{p_o^w+1}{2}\right) -7.09\sqrt{n_z\log_2\frac{2}{\varepsilon_1^2}}\\
&-n_z^dH(p_a),
	\end{aligned}
\end{equation}
where $p_o^w=\sqrt{4({p_x^w}^2+{p_y^w}^2+{p_z^w}^2-p_x^w-p_y^w-p_z^w)+3}$.
%where $p_a$ is given by Eq.~\eqref{Eq:pa}.

\subsubsection{Discard method}
In practice, the random assignment might be technically challenging to implement. Thus, we consider an simpler method to deal with the double-click events in the $Z$ basis --- discarding all the double-click events. Denote $\rho^s$ and $\rho^d$ to be the density matrices of the squashing qubit states, single-click, and double-click, respectively. Then, one has $p_d\rho^d_A+(1-p_d)\rho^s_A=\rho_A$, where $p_d=n_z^d/n_z$ is the ratio of double-click events in the $Z$ basis. Once discarding all the double-click events, the random bits in the remaining data can be lower bounded by
\begin{equation}\label{Eq:DisRtomo1}
R^{n}_{t}\ge n_z^s C(\rho_A^s) -7.09\sqrt{n_z\log_2\frac{2}{\varepsilon_1^2}},
\end{equation}
where $n_z^s=n_z-n_z^d$ is the number of single-click events in the $Z$ basis. To estimate $C(\rho_A^s)$, one can employ the concavity of relative entropy of coherence,
\begin{equation}
p_dC(\rho^d_A)+(1-p_d)C(\rho^s_A)\ge C(\rho_A).
\end{equation}
Thus
\begin{eqnarray}\label{Eq:CohIneq}
n_z^s C(\rho^s_A)&\ge& n_zC(\rho_A) - n_dC(\rho^d_A)\\\nonumber
&\ge& n_zC(\rho_A) - n_d.
\end{eqnarray}
where $C(\rho^d_A)\le 1$ is used in the second inequality. Combining Eqs.~\eqref{Eq:CohIneq}, \eqref{Eq:DisRtomo1}, and \eqref{Eq:RH}, the net randomness generation rate is given by
\begin{equation}\label{Eq:DisRtomo2}
\begin{aligned}
R^{n}_{t}\ge \, & n_zH(p_z^w) - n_zH\left(\frac{p_o^w+1}{2}\right) \\
&-7.09\sqrt{n_z\log_2\frac{2}{\varepsilon_1^2}}-n_z^d
\end{aligned}
\end{equation}
Note that the randomness generation rate by the discard method is generally smaller than that the rate by the random assignment method in Eq.~\eqref{Eq:AssRtomo}.

\subsection{Analysis of statistic fluctuations} \label{Sec:Finite}
In practice, the number of squashed qubits $n$ is finite. Thus the probabilities used for parameter estimation would suffer from statistical fluctuations. Here, we analyze the finite-data-size effect on the estimation of $p^w_j$. In order to distinguish the probabilities with the measurement rates, denote the expectation values of $p^w_j$ to be $\bar{p}_j$, which would be inserted into Eq.~\eqref{Eq:Rtomo} to evaluate the genuine randomness. Since the qubits are assumed to be i.i.d., Alice can employ the Hoeffding inequality \cite{hoeffding1963probability} to estimate the discrepancy between $p^w_j$ and $\bar{p}_j$ caused by the statistical fluctuations,
\begin{equation}
\textrm{Prob}(\bar{p}_j  \le p_j^w - \theta) \le e^{-2 \theta^2n_j} = \varepsilon_j,\label{Eq:ERX}
%\textrm{Prob}(\bar{p}_y  \le p_y^L - \theta) &\le e^{-2 \theta^2 n_y} = \varepsilon_y,\label{Eq:ERY}\\
%\textrm{Prob}(\bar{p}_z  \le p_z^L - \theta) &\le e^{-2 \theta^2 n_z} = \varepsilon_z.\label{Eq:ERZ}
\end{equation}
Here, we assume $p^w_j - \theta \ge 1/2$, otherwise we take the worst bound $\bar{p}_j=1/2$. Replacing $\{p^w_j\}$ by $\{p^w_j - \theta\}$ in Eq.~\eqref{Eq:AssRtomo} or \eqref{Eq:DisRtomo2}, one can obtain the lower bound of randomness with a total failure probability,
\begin{equation}
\varepsilon = \varepsilon_1 + \varepsilon_2 + \varepsilon_x + \varepsilon_y + \varepsilon_z,
\end{equation}
where $\varepsilon_1$ is introduced by the smooth parameter and $\varepsilon_2$ is introduced by randomness extraction.

\section{Simulation}\label{Sec:Simulation}
In this section, we first analyze the performance of the tomography-based SIQRNG and compare it to the original proposal. Then, by taking account of statistical fluctuations, we optimize the ratio of qubits used for the tomography testing.

\subsection{Comparing to the original SIQRNG}
For simplicity, we consider the asymptotic limit where the number of emitted quantum states $N$ is infinitely large. Then, the randomness generation rate for the original protocol is given by Eq.~\eqref{Eq:UCGR}, whereas the rate for the tomography-based protocol is given by Eq.~\eqref{Eq:TMGR}. In both cases, $q_z$ is set to be $1$. Besides, we assume the single photon source is used without considering the photon loss and the detector inefficiency, and hence the transmittance $\beta=1$. Then, the randomness generation rate for the original protocol is given by
\begin{equation} \label{Eq:ruIdeal}
r_u|_{q_z=1,\beta=1} \ge 1-H(\Delta_X(\rho_A)),
\end{equation}
and the tomography-based protocol by
\begin{equation} \label{Eq:rtIdeal}
r_t|_{q_z=1,\beta=1} \ge C(\rho_A).
\end{equation}

In the comparison, we assume the input state has the form of $\rho_A=(\mathrm{I}+x\sigma_x+y\sigma_y)/2$, where $x$ and $y$ are two parameters and $x^2+y^2\le1$. The comparison between Eqs.~\eqref{Eq:ruIdeal} and \eqref{Eq:rtIdeal} is illustrated in Fig.~\ref{fig:COMPARISON}. One can clearly see that the tomography-based scheme generally provides a higher randomness generation rate than the original proposal. The larger the parameter $y$ is, the bigger gaps the two schemes have. In general, one can consider a more general state, $\rho_A=(\mathrm{I}+x\sigma_x+y\sigma_y+z\sigma_z)/2$, where the gap of randomness generation rate between the two schemes is nonzero as long as $y \neq 0$.

\begin{figure}
\begin{minipage}[t]{0.5\linewidth}
\centering
\includegraphics[width=2.5in]{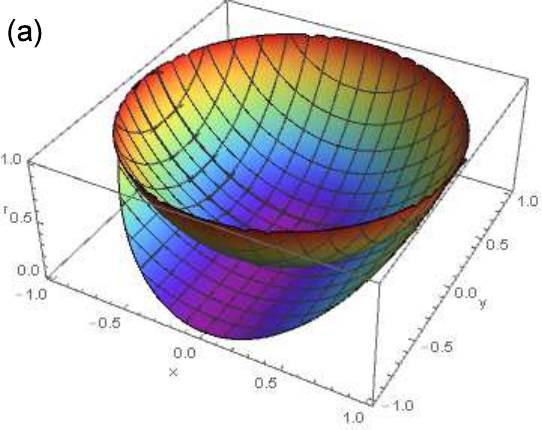}
%\caption*{(a)}
\label{fig:side:a}
\end{minipage}%
\begin{minipage}[t]{0.5\linewidth}
\centering
\includegraphics[width=2.5in]{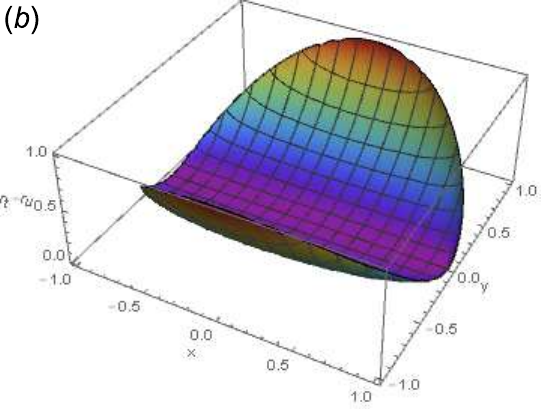}
%\caption*{(b)}
\label{fig:side:b}
\end{minipage}
\caption{Comparison of the randomness generation rates with an input qubit state $\rho_A=\frac{\mathrm{I}+x\sigma_x+y\sigma_y}{2}$, with $N\rightarrow\infty, q_z=1, \beta=1$. (a) The lower surface describes the randomness generation rate for the uncertainty-relation-based scheme $r_u$ as shown in Eq.~\eqref{Eq:ruIdeal}, whereas the upper surface describes the randomness generation rate for the tomography-based scheme $r_t$ as shown in Eq.~\eqref{Eq:rtIdeal}. (b) Illustration of the gap between the two schemes, $r_t-r_u \ge 0$.}
\label{fig:COMPARISON}
\end{figure}

\subsection{Parameter optimization}
Now we analyze the performance of the tomography-based protocol by simulating a practical experiment setup. Details of the simulation model is presented in Appendix \ref{Ap:Simulation}. Consider a practical source, consisting of a laser and a polarization modulator, which emits $N$ coherent-state pulses with an intensity of $\mu_0$. We assume the quantum state is prepared to be $\ket{+}$, which is then transmitted through a depolarization channel. Thus, the received quantum state can be described by
\begin{equation} \label{Eq:rhodepol}
\rho_A=p\frac{\mathrm{I}}{2}+(1-p)\ketbra{+}{+},
\end{equation}
where $p\in[0, 1]$. The readout system consists of a polarization rotator (used for basis selection), a polarization beam splitter, and two threshold detectors with the same detection efficiencies. Denote the total transmittance by $\eta$, including the detector efficiency and the coupling efficiency between the source and the detector. It is equivalent to consider a coherent state with intensity of $\mu\equiv\mu_0\eta$. Here, we ignore the detection caused by the dark count since for QRNG dark counts are normally negligible comparing to $\eta$. For a more comprehensive model taking account of the dark counts, one can refer to the corresponding QKD model \cite{Ma2005Practical}.

In this model, we put the misalignment errors into the parameter $p$. Then, the simulated worst estimations of $p_j$ are given by,
\begin{equation} \label{Eq:OPT2}
	\begin{aligned}
		\bar{p}_x &= \frac{1-p-e^{-\mu}+pe^{-\frac{\mu}{2}}}{1-e^{-\mu}} - \theta,\\
\bar{p}_y &= \frac{e^{-\frac{\mu}{2}}-e^{-\mu}}{1-e^{-\mu}} - \theta,\\
%\bar{p}_z = \frac{1-e^{-\frac{\mu}{2}}}{(1-e^{-\mu})(1-\delta)}.
\bar{p}_z &=\frac{e^{-\frac{\mu}{2}}-e^{-\mu}}{1-e^{-\mu}}-\theta.
	\end{aligned}
\end{equation}
Meanwhile, the number of double click events in the $Z$ basis is
\begin{equation}\label{Eq:OPTd}
n^{d}_z=N(1-2q)(1+e^{-\eta\mu_0}-2e^{-\frac{\eta\mu_0}{2}}).
\end{equation}
The detailed model with the calculations of Eqs.~\eqref{Eq:OPT2} and \eqref{Eq:OPTd} are shown in Appendix \ref{Ap:Simulation}. Note that all the $\bar{p}_j$ here are less than $1/2$, as required in the randomness analysis shown in Sec.~\ref{Sec:protocol}. Then one can evaluate the extractable randomness from Eq.~\eqref{Eq:DisRtomo2}. Here, we adopt the discard method to process the double-click events in the $Z$ basis. In fact, one would obtain the same randomness generation rate by the random assignment method in the case of Eq.~\eqref{Eq:rhodepol}.

In the simulation, we pick $p=0$, $p=0.1$, $p=0.3$ for the input state of Eq.~\eqref{Eq:rhodepol}, and set $\varepsilon_x=\varepsilon_y=\varepsilon_z=\varepsilon_1=\varepsilon_2=10^{-10}$ and the number of pulses $N=10^{10}$. First, by optimizing the tomography testing parameter $q$, we show the dependence of the randomness generation rate, given by Eq.~\eqref{Eq:DisRtomo2}, on the intensity $\mu$ in Fig.~\ref{fig:NUMERIC1}. From the figure, one can see that, initially, the randomness generation rate increases with the intensity of the signal $\mu$ due to the increase in the single-click events relative to the no-click events. As $\mu$ keeps increasing, the double-click events become dominant. Then, the randomness generation rate starts to decrease. Thus, there is an optimal choice of $\mu$. In experiment, one should characterize total transmittance $\eta$ first and then set the light intensity $\mu_0$ to make $\mu=\mu_0\eta$ optimal.

\begin{figure}[ht!]
\centering
\resizebox{9cm}{!}{\includegraphics[scale=1]{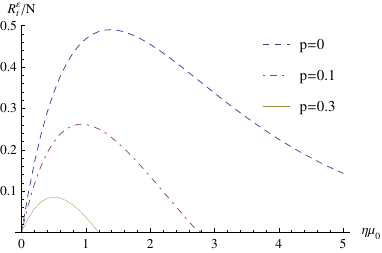}}
\caption{Randomness generation rate vs.~intensity $\mu$ with various depolarization parameters $p$. The optimal $\mu$ for $p=0$, $0.1$, and $0.3$ are $\mu=1.4$, $0.9$, and $0.5$, respectively.}\label{fig:NUMERIC1}
\end{figure}

Next, we investigate how the optimal tomography testing parameter, $q$, varies with the number of pulses, $N$. Here, we pick $p=0.1$ and optimize the the intensity $\mu$. When $N\le10^{4.8}$, no net random bits can be generated. One can see from Fig.~\ref{fig:NUMERIC2}(a) that the optimal $q$ starts from $0.14$ and drops down close to 0 with the increase in $N$. This is consistent with the intuition that $q_z=1-2q\rightarrow 1$ as $N$ goes to infinity. Note that the optimization of $q$ assume that $q_x=q_y$ in the tomography-based protocol. In general, one can optimize $q_x$ and $q_y$ separately. %The results are similar\footnote{Typically, the number of pulses is $N\ge 10^6$, where $q_x+q_y=2q\le 0.1$, from Fig.~\ref{fig:NUMERIC2}(a), with $q$ as the optimal value. Suppose one could increase the randomness generation rate significantly by biasing $q_x$ and $q_y$, then one can also increase the rate by changing $q_x=q_y=q$ in the current protocol, which is however inconsistent with the fact that $q$ is optimal.}.

We also investigate how the randomness generation rate varies with the number of pulses $N$, as shown in Fig.~\ref{fig:NUMERIC2}(b). Here, we pick $p=0.1$ and optimize both the intensity $\mu$ and the tomography testing parameter $q$. One can see that no randomness can be obtained as $N\le10^{4.8}$, beyond which point the rate increases with $N$. This increase is mainly contributed by the increase ratio of the $Z$-basis measurement as $N$ becomes large. This is similar to the biased basis case in QKD \cite{Fung2010Finite}.

\begin{figure}[ht!]
\begin{minipage}[t]{0.5\linewidth}
\centering
\includegraphics[width=2.5in]{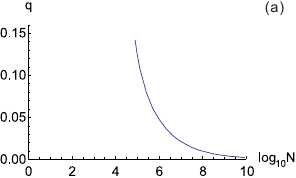}
%\caption*{(a)}
\label{fig:simulation:a}
\end{minipage}%
\begin{minipage}[t]{0.5\linewidth}
\centering
\includegraphics[width=2.5in]{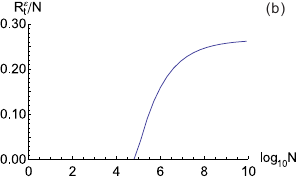}
%\caption*{(b)}
\label{fig:simulation:b}
\end{minipage}
\caption{The performance of the tomography-based SIQRNG by optimizing the basis selection parameter $q$, where we pick $p=0.1$ and have the intensity parameter $\mu$ optimized. (a) shows the optimal value of $q$ vs the number of pulses $N$. (b) illustrates the randomness generate rate $R^{\varepsilon}_t/N$ vs the number of pulses $N$.}
\label{fig:NUMERIC2}
\end{figure}

\section{Discussion}\label{Sec:Discussion}
In this paper, we propose a framework for SIQRNG via measuring coherence of an unknown quantum state. We show that the uncertainty-relation-based SIQRNG is essentially related to estimating the relative entropy of coherence with a coherence witness. Furthermore, we propose a SIQRNG scheme based on state tomography. By simulating a typical QRNG setup, we show that our protocol generally enjoys a higher randomness generation rate than the uncertainty-relation-based ones.

The security analysis of QRNG is very similar to that of QKD. The mathematical definition of security in two tasks is essentially the same. For example, privacy amplification in QKD is closely related to the randomness extraction in QRNG. In practice, there are mainly two differences between them. (a) QKD involves two legitimate parties Alice and Bob; thus it requires error correction to ensure the consistency of the random numbers shared between them; whereas QRNG only involves one party, and hence error correction is unnecessary. (b) Local randomness used for encoding and basis selection is free in QKD, whereas randomness is a resource in QRNG. Nevertheless, most of the security analysis techniques in QRNG, including ours, are borrowed from QKD. In terms of the security analysis, the i.i.d.~assumption in QRNG is equivalent to the collective attack assumption in QKD, whereas the non-i.i.d.~scenario in QRNG corresponds to the coherent attack assumption in QKD. The difference between the two attacks vanishes when the data size goes to infinity \cite{christandl2009postselection}, and we would expect the same deduction for our framework. In order to link coherence with randomness in the non-i.i.d.~case, one needs to consider the one-shot coherence distillation \cite{zhao2018one,bu2017max}. This is an interesting subject for future study.

Along the direction of this paper, one can realize a SIQRNG scheme by designing the coherence witness that is adapted to specific experimental conditions. Besides, it is promising to extend the framework to high-dimensional QRNG, e.g., schemes based on continuous variables \cite{Marangon17SI} or laser phase fluctuations \cite{ZhouRandomness}. A possible challenge of this extension is the development of the high-dimensional squashing model. As the resource framework of coherence is related to the security proof of QRNG, it is also interesting to investigate whether similar relation exists between coherence and QKD.

\section*{Acknowledgments}
We thank Zhen Zhang for fruitful discussions. This work was supported by the National Natural Science Foundation of China Grants No.~11674193 and No.~11875173 and the National Key R\&D Program of China Grants No.~2017YFA0303900 and No.~2017YFA0304004.

\appendix
\section{Smooth min-entropy}\label{Ap:Conditional}
In this appendix, we provide the definition of smooth min-entropy \cite{konig2009operational}.
\begin{definition}
Given a bipartite density operator $\rho_{AE}$, the min-entropy of $A$ conditioned on $E$ is defined as
\begin{equation}
H_{\mathrm{min}}(A|E)_{\rho_{AE}}\equiv -\min_{\sigma_E}D_{\infty}(\rho_{AE}||\mathrm{I}\otimes\sigma_E)
\end{equation}
where the minimization ranges over all normalized density operators $\sigma_E$ on $E$ and
\begin{equation}
D_{\infty}(\tau||\tau')\equiv \min\{\lambda\in \mathbb{R} : \tau\leq 2^{\lambda}\tau'\}.
\end{equation}
\end{definition}
\noindent Then the smooth min-entropy of $A$ conditioned on $E$ is defined as
\begin{equation}
H_{\mathrm{min}}^{\varepsilon}(A|E)_{\rho_{AE}}\equiv \sup_{\rho_{AE}'}H_{\mathrm{min}}(A|E)_{\rho_{AE}'},
\end{equation}
where the supremum ranges over all density operators $\rho_{AE}'$ which are $\varepsilon$-close to $\rho_{AE}$. Normally, the distance between $\rho_{AE}'$ and $\rho_{AE}$ can be measured by Bures distance $\|\rho-\sigma\|_B=\sqrt{2-F(\rho,\sigma)}$, where $F(\rho,\sigma)=\|\sqrt{\rho}\sqrt{\sigma}\|_1$ and $\|\cdot\|_1$ is the $l_1$-norm.

\section{Derivation of Eq.~\eqref{Eq:RNC}}\label{Ap:Derivation}
With the i.i.d.~assumption, the amount of randomness from transmitted quantum states is given by
\begin{equation}\label{Eq:RHM2}
R^{\varepsilon_1}(K_z) = \min_{\rho_{AE}}H_{\mathrm{min}}^{\varepsilon_1}(A|E)_{\Delta_{Z^{\otimes n_z}}^A(\rho_{AE}^{\otimes n_z})}.
\end{equation}
As $n_z\ge \frac85\log_2\frac{2}{\varepsilon_1^2}$, the smooth min-entropy can be lower bounded by the conditional von Neumann entropy , $H(A|E)_{\rho_{AE}}=H(\rho_{AE})-H(\rho_E)$ \cite{tomamichel2009fully},
\begin{equation}\label{Eq:HMN}
\begin{aligned}
	H_{\mathrm{min}}^{\varepsilon_1}(A|E)_{\Delta_{Z^{\otimes n_z}}^A(\rho_{AE}^{\otimes n_z})}\ge n_zH(A|E)_{\Delta_Z^A(\rho_{AE})}-\sqrt{n_z}\delta(\varepsilon_1, \eta),
\end{aligned}
\end{equation}
where $\delta(\varepsilon_1, \eta)=4(\log_2\eta)\sqrt{\log_2\frac{2}{\varepsilon_1^2}}$ and $\eta\le\sqrt{2^{-H_{\mathrm{min}}(A|E)_{\Delta_Z^A(\rho_{AE})}}}+\sqrt{2^{H_{\mathrm{max}}(A|E)_{\Delta_Z^A(\rho_{AE})}}}+1$, with  $H_{\mathrm{min}}$ and $H_{\mathrm{max}}$ being the min-entropy and maximal-entropy, respectively. Here, $H_{\mathrm{min}}(A|E)_{\Delta_Z^A(\rho_{AE})}\ge 0$ and $H_{\mathrm{max}}(A|E)_{\Delta_Z^A(\rho_{AE})}\le 1$ and hence one has $\eta\le 2+\sqrt{2}$. Thus $\delta(\varepsilon_1, \eta)\le k\sqrt{\log_2\frac{2}{\varepsilon_1^2}}$, with $k=4\log_2(2+\sqrt{2})\approx 7.09$. Therefore, inserting Eq.~\eqref{Eq:HMN} into Eq.~\eqref{Eq:RHM2}, one has
\begin{equation}\label{Eq:RNC3}
R^{\varepsilon_1}(K_z) \ge n_z\min_{\rho_{AE}}H(A|E)_{\Delta_Z^A(\rho_{AE})}-7.09\sqrt{n_z\log_2\frac{2}{\varepsilon_1^2}}.
\end{equation}

\section{Partial derivatives of $C(\rho_A)$}\label{Ap:Partial}
In this appendix, we analyze the partial derivatives of $C(\rho_A)$. From Eq.~\eqref{Eq:RH},
\begin{equation}
C(\rho_A) = H(p_z)- H\left(\frac{p_o+1}{2}\right),
\end{equation}
where $p_o= \sqrt{4(p_x^2+p_y^2+p_z^2-p_x-p_y-p_z)+3}$. Thus,
\begin{equation}
	\begin{aligned}
		 \frac{\partial C(\rho_A)}{\partial p_z} &= \frac{\partial}{\partial p_z}\left[H(p_z)-H\left(\frac{p_o+1}{2}\right)\right]		\\
&= \frac{2p_z-1}{p_o}\log_2\left(\frac{1+p_o}{1-p_o}\right)-\log_2\left(\frac{p_z}{1-p_z}\right),  \\
 \frac{\partial C(\rho_A)}{\partial p_x} &= \frac{2p_x-1}{p_o}\log_2\left(\frac{1+p_o}{1-p_o}\right),	 \\
\frac{\partial C(\rho_A)}{\partial p_y} &= \frac{2p_y-1}{p_o}\log_2\left(\frac{1+p_o}{1-p_o}\right).
	\end{aligned}
\end{equation}
By analyzing above equations, for $j\in\{x, y, z\}$,  $\frac{\partial^2 C(\rho_A)}{\partial p_j^2} \ge 0$. Therefore, function $\frac{\partial C(\rho_A)}{\partial p_j}$ is nondecreasing with $p_j$. Thus, %$\frac{\partial C(\rho_A)}{\partial p_x} , \frac{\partial C(\rho_A)}{\partial p_y} = 0$ for $p_j = 1/2$ where .
\begin{eqnarray}
\frac{\partial C(\rho_A)}{\partial p_j} &=& 0\hspace{2cm}p_j = 1/2, \nonumber \\
&\le& 0 \hspace{2cm}		p_j \le 1/2, \nonumber \\
&\ge& 0 \hspace{2cm}		p_j \ge 1/2.    \nonumber
\end{eqnarray}
%\section{Sim}
%In this section, we prove that $\eta\le 2+\sqrt{2}$. From Theorem 9 in Ref.~\cite{tomamichel2009fully}, $\eta=\Upsilon(Z|E)_{\rho_{AE}|\rho_{AE}}\le\sqrt{2^{-H_{\mathrm{min}}(Z|E)_{\rho_{AE}}}}+\sqrt{2^{-H_{\mathrm{max}}(Z|E)_{\rho_{AE}}}}+1$, where $H_{\mathrm{min}}$ and $H_{\mathrm{max}}$ are the min-entropy and maximal-entropy, respectively. Since $H_{\mathrm{min}}(Z|E)_{\rho_{AE}}\ge 0$ and $H_{\mathrm{max}}(Z|E)_{\rho_{AE}}\le 1$, one has $\eta\le 2+\sqrt{2}$.

\section{Simulation model}\label{Ap:Simulation}
In this appendix, we analyze a simulation model with a practical experimental setup. Consider a practical source, consisting of a laser and a polarization modulator, which emits $N$ coherent-state pulses with an intensity of $\mu_0$. We assume the quantum state is prepared to be $\ket{+}$, which is then transmitted through a depolarization channel. Thus, the received quantum state can be described by $\rho_A=p\frac{\mathrm{I}}{2}+(1-p)\ketbra{+}{+}$, where $p\in[0, 1]$. The readout system consists of a polarization rotator (used for basis selection), a polarization beam splitter, and two threshold detectors with the same detection efficiencies. Denote the total transmittance by $\eta$, including the detector efficiency and the coupling efficiency between the source and the detector. Here, we ignore the detection caused by dark counts since for QRNG, dark counts are normally negligible comparing to $\eta$.

In the following, we aim to evaluate the expected value of all the directly obtained experimental statistics. They include the number of total click events, $n_j$, the number of single-click events of the two detectors, $n_j^{0}$, $n_j^{1}$, and the number of double-click events, $n_j^{d}$, when measuring in the $j\in\{X, Y, Z\}$ basis.

Note that the number of photons of the coherent-state pulse follows the Poisson distribution, $P(n)=\frac{e^{-\mu_0}\mu_0^n}{n!}$. With the total transmittance of the system $\eta$, the total number of clicks in the $X$ basis is
\begin{eqnarray}
n_x&=&Nq\sum_n e^{-\mu_0}\frac{\mu_0^n}{n!}[1-(1-\eta)^n]\nonumber\\
&=&Nq(1-e^{-\eta\mu_0}),
\end{eqnarray}
where $q$ is the probability of selecting the $X$ basis. Similarly, one has
\begin{eqnarray}
n_y&=&Nq(1-e^{-\eta\mu_0}),\\
n_z&=&N(1-2q)(1-e^{-\eta\mu_0}),
\end{eqnarray}
Denote the probability of double clicks when emitting $m$ photons and measuring in the $j\in\{X, Y, Z\}$ basis by $p_{doub}^{i,m}$. Since the polarization of the input state is $p\frac{\mathrm{I}}{2}+(1-p)\ketbra{+}{+}$, in which only the component $\frac{\mathrm{I}}{2}$ may result in the double-click events, then,
\begin{eqnarray}\label{Eq:AppB16}
p_{doub}^{x,m}&=&\frac{p}{2^m}\sum_k C_m^k[1-(1-\eta)^k][1-(1-\eta)^{m-k}]\nonumber\\
&=& \frac{p}{2^m}\sum_k [C_m^k-C_m^k(1-\eta)^k-C_m^k(1-\eta)^{m-k}+C_m^k(1-\eta)^m]\nonumber\\
&=& \frac{p}{2^m}[2^m+(1-\eta)^m2^m-2(2-\eta)^m]\nonumber\\
&=& p(1+(1-\eta)^m-2(1-\frac\eta2)^m).
\end{eqnarray}
Meanwhile, for measurement basis $Y$ and $Z$, both component $\frac{\mathrm{I}}{2}$ and $\ketbra{+}{+}$ of $\rho_A$ result in the double-click events with equal probability, thus one has
\begin{eqnarray}\label{Eq:AppB2}
p^{y, m}_{doub} =p^{z,m}_{doub} &=& \frac{1}{2^m}\sum_k C_m^k[1-(1-\eta)^k][1-(1-\eta)^{m-k}]\nonumber\\
&=&1+(1-\eta)^m-2(1-\frac\eta2)^m.
\end{eqnarray}
Thus, the total number of double clicks in the $X$ basis is
\begin{eqnarray}
n_x^d&=&Nq\sum_m e^{-\mu_0}\frac{\mu_0^m}{m!}p^{x,m}_{doub}\nonumber\\
&=&Nqp(1+e^{-\eta\mu_0}-2e^{-\frac{\eta\mu_0}{2}}).
\end{eqnarray}
And similarly, one has
\begin{eqnarray}
n_y^d&=&Nq\sum_m e^{-\mu_0}\frac{\mu_0^m}{m!}p^{y,m}_{doub}\nonumber\\
&=&Nq(1+e^{-\eta\mu_0}-2e^{-\frac{\eta\mu_0}{2}}),\\
n_z^d&=&N(1-2q)\sum_m e^{-\mu_0}\frac{\mu_0^m}{m!}p^{z,m}_{doub}\nonumber\\
&=&N(1-2q)(1+e^{-\eta\mu_0}-2e^{-\frac{\eta\mu_0}{2}}),
\end{eqnarray}
Now we evaluate the number of single-click events corresponding to outcome $\ket{+}$ and $\ket{-}$. Note that, for $\rho_A$, the component $\ketbra{+}{+}$ never results in outcome $\ket{-}$, whereas the component $\frac{\mathrm{I}}{2}$ contributes to the single-click events of $\ket{+}$ and $\ket{-}$ with the equal probability. Thus, one has
\begin{equation}
\begin{aligned}
n_x^{0}&=(1-p)n_x+\frac12(pn_x-n^d_x) \\
&=Nq(1-p-e^{-\eta\mu_0}+pe^{-\frac{\eta\mu_0}{2}}).\\
n_x^{1}&=\frac12(pn_x-n^d_x) \\
&=Nqp(e^{-\frac{\eta\mu_0}{2}}-e^{-\eta\mu_0}).
\end{aligned}
\end{equation}
For measurement basis $Y$ and $Z$, both the component $\ketbra{+}{+}$ and $\frac{\mathrm{I}}{2}$ contributes to the single-click events of the two outcomes with the equal probability. Thus, one has
\begin{equation}
\begin{aligned}
n_y^{0}=n_y^{1}&=\frac12(n_y-n_y^d)\\
&=Nq(e^{-\frac{\eta\mu_0}{2}}-e^{-\eta\mu_0}).\\
n_z^{0}=n_z^{1}&=\frac12(n_z-n_z^d)\\
&=N(1-2q)(e^{-\frac{\eta\mu_0}{2}}-e^{-\eta\mu_0}).
\end{aligned}
\end{equation}
Combining the above results, the experimental obtained statistics are given by
\begin{equation}
\begin{aligned}
n^{0}_x&=Nq(1-p-e^{-\eta\mu_0}+pe^{-\frac{\eta\mu_0}{2}}),\\
n^{1}_x&=Nqp(e^{-\frac{\eta\mu_0}{2}}-e^{-\eta\mu_0}),\\
n^{d}_x&=Nqp(1+e^{-\eta\mu_0}-2e^{-\frac{\eta\mu_0}{2}}),\\
n^{0}_y&=n^{1}_y=Nq(e^{-\frac{\eta\mu_0}{2}}-e^{-\eta\mu_0}),\\
n^{d}_y&=Nq(1+e^{-\eta\mu_0}-2e^{-\frac{\eta\mu_0}{2}}),\\
n^{0}_z&=n^{1}_z=N(1-2q)(e^{-\frac{\eta\mu_0}{2}}-e^{-\eta\mu_0}),\\
n^{d}_z&=N(1-2q)(1+e^{-\eta\mu_0}-2e^{-\frac{\eta\mu_0}{2}}),\\
n_x&=Nq(1-e^{-\eta\mu_0}),\\
n_y&=Nq(1-e^{-\eta\mu_0}),\\
n_z&=N(1-2q)(1-e^{-\eta\mu_0}).
\end{aligned}
\end{equation}
Inserting these expressions into Eq.~\eqref{Eq:Squash}, one has
\begin{equation}
\begin{aligned}
\frac{1-p-e^{-\eta\mu_0}+pe^{-\frac{\eta\mu_0}{2}}}{1-e^{-\eta\mu_0}}\le & p_x \le \frac{1-(1-p)e^{-\eta\mu_0}-pe^{-\frac{\eta\mu_0}{2}}}{1-e^{-\eta\mu_0}},\\
\frac{e^{-\frac{\eta\mu_0}{2}}-e^{-\eta\mu_0}}{1-e^{-\eta\mu_0}}\le & p_y \le \frac{1-e^{-\frac{\eta\mu_0}{2}}}{1-e^{-\eta\mu_0}},\\
\frac{e^{-\frac{\eta\mu_0}{2}}-e^{-\eta\mu_0}}{1-e^{-\eta\mu_0}}\le & p_z \le \frac{1-e^{-\frac{\eta\mu_0}{2}}}{1-e^{-\eta\mu_0}}.
\end{aligned}
\end{equation}
Following the analysis of the squashing model and statistical fluctuations, the worst case expectation values of $p_x$, $p_y$ and $p_z$ are thus given by
\begin{equation}
\begin{aligned}
\bar{p}_x = P_x^L - \theta &= \frac{1-p-e^{-\eta\mu_0}+pe^{-\frac{\eta\mu_0}{2}}}{1-e^{-\eta\mu_0}} - \theta,\\
\bar{p}_y = P_y^L - \theta &= \frac{e^{-\frac{\eta\mu_0}{2}}-e^{-\eta\mu_0}}{1-e^{-\eta\mu_0}} - \theta,\\
\bar{p}_z = P_z^L - \theta &= \frac{e^{-\frac{\eta\mu_0}{2}}-e^{-\eta\mu_0}}{1-e^{-\eta\mu_0}}-\theta.
\end{aligned}
\end{equation}
Inserting $\bar{p}_j$ into Eq.~\eqref{Eq:AssRtomo} or \eqref{Eq:DisRtomo2}, one can estimate the amount of extractable randomness from the raw random bits.

%\section*{References}
%%%%%%%%%%%%%%%%%%%%%%%%%%%%%%%%%%%%%%%%
% choose a style
%\bibliographystyle{ieeetr}
% \bibliographystyle{iopart-num}
\bibliographystyle{apsrev4-1}
%%%%%%%%%%%%%%%%%%%%%%%%%%%%%%%%%%%%%%%%

%%%%%%%%%%%%%%%%%%%%%%%%%%%%%%%%%%%%%%%%
% choose a .bib file
\bibliography{bibCR}
%%%%%%%%%%%%%%%%%%%%%%%%%%%%%%%%%%%%%%%%
\end{document}